\documentclass[review,3p,authoryear]{elsarticle}
\usepackage{amsmath,amssymb,amsfonts,amsthm}
\usepackage{enumerate}
\usepackage{float}
\usepackage{graphicx}
\usepackage{array,multirow,makecell}
\usepackage{natbib}
\usepackage{multicol} 
\usepackage{lscape} 
\usepackage{pdflscape}
\usepackage{diagbox}
\usepackage {rotating}
\usepackage{threeparttable}
\newtheorem{Theorem}{Theorem}

\newtheorem{Lemma}{Lemma}

\newtheorem{Remark}{Remark}
\newproof{pot1}{Proof of Theorem \ref{the1}}
\newproof{pot2}{Proof of Theorem \ref{the2}}

\newcommand{\reftab}[1]{Table \ref{#1}}
\newcommand{\reffig}[1]{Figure \ref{#1}}
\newcommand{\reflem}[1]{Lemma \ref{#1}}
\usepackage{color}
\begin{document}
\title{A closed-form approximation for pricing geometric Istanbul options}
\author[uc1]{Mohamed Amine Kacef \corref{cor}}
\ead{kacefmohamedamine@gmail.com}
\author[uc1]{Kamal Boukhetala }
\ead{kboukhetala@usthb.dz}
\address[uc1]{Department of Probability and Statistics, Faculty of Mathematics, University of Sciences and Technology, Houari Boumediene USTHB, BP 32, El-Alia, Bab Ezzouar 16111, Algiers, Algeria.}
\cortext[cor]{Corresponding author}

\begin{abstract}
The Istanbul options were first introduced by Michel Jacques in 1997. These derivatives are considered as an extension of the Asian options. In this paper, we propose an analytical approximation formula for a geometric Istanbul call option (GIC) under the Black-Scholes model. Our approximate pricing formula is obtained in closed-form using a second-order Taylor expansion. We compare our theoretical results with those of Monte-Carlo simulations using the control variates method. Finally, we study the effects of changes in the price of the underlying asset on the value of GIC.
\end{abstract}
\maketitle
\section{Introduction}
The Istanbul option (IO) is an exotic option whose payoff depends on whether the price of the underlying asset has reached or not a certain threshold previously fixed named $\textit{barrier}$. If this barrier is reached before maturity, an Asian option (AO) is activated and the average is calculated from the first moment when the price of the underlying asset reaches the barrier until the maturity. However, if the barrier is not reached, a standard European option (EO) is activated at maturity. The IO can therefore be seen as an hybrid option that has the characteristics of both AO and EO. This option is also similar to the AO with barrier studied by \citet{KK} and \citet{6327776}, the main difference being that in IO the calculation of the average is activated from the first hitting time of the barrier and not from the acquisition date of the contract. 

If we consider a \citet{BlackScholes11} model, the valuation of products such as the arithmetic Asian options (AAOs) becomes very difficult since the hypothesis of taking the underlying asset price is a geometric Brownian motion does not allow to obtain a closed-form pricing formula because the distribution of the sum of log-normal random variables is not known in theory. However, the price of AAO can be approximated in practice by Monte-Carlo (MC) simulations with variance reduction techniques (see \citet{Zhang_pricingasian}, \citet{CC} and \citet{LU2019}). It is also possible to approach the price of AAO with a Taylor expansion as in \citet{JU}.  

For the geometric Asian option (GAO), the pricing formula is known in closed-form (see \citet{Kemna} for the call and \citet{Angus} for more examples of payoffs). Recently, the GAOs with barrier have been studied by \citet{Aimi3} and \citet{Aimi2}. The price of this type of options has no closed-form expression and is increasingly the subject of financial research. The options involving a geometric average are also studied in the context of stochastic volatility (for examples, see \citet{doi:10.1088/1469-7688/4/3/006} and \citet{HUBALEK20113355}).

In \citet{MichelJac1}, the arithmetic Istanbul call option (AIC) is study in continuous and discrete time trading. The price of AIC is obtained through a log-normal approximation with the moment-matching method (for more details on this approach, see \citet{LEVY1992474}). In this article, we focus our attention on the pricing problem of the geometric Istanbul option (GIO) in continuous time trading. We consider only the case of a call option with an up-barrier and a fixed strike price. We also suppose that the terms of the contract do not guarantee any payment of dividend or rebate at maturity.

This article is organized as follows. In section \ref{section2}, we describe the continuous-time economic model chosen for our study and its theoretical properties. In section \ref{section3}, we show with the strong Markov property that the price formula of GIC can be written in semi-closed form. Then, we propose an analytic approximation formula of this price using a second-order Taylor expansion. In section \ref{section4}, we compare our theoretical results with those of MC simulations using the control variates (CV) method to reduce the variance of MC estimator. We also compare the price of GIC with AIC, and analyze the price sensitivity of GIC to changes in the price of the underlying asset. Finally, in section \ref{section5}, we conclude with a summary of the main results obtained in this article.
\section{Financial model description}\label{section2}
We consider a standard Black and Scholes model of frictionless markets where there is no arbitrage opportunity, the risk-free interest rate $r$ and volatility $\sigma>0$ are constant. The underlying stock price  $S_{t}$ follows a geometric Brownian motion 
\begin{equation}
{S_t} = {S_0}\exp \left( {\overline{\mu} t + \sigma W_t} \right), \ \     t\in [0,T],
\label{GBM}
\end{equation}
where $[0,T]$ is the trading period, ${S_0>0}$ is the initial stock price, $\overline{\mu}= r-\sigma^2/2$ is the risk-neutral drift rate and $W_t$ a one-dimensional standard Brownian motion under the risk-neutral probability $\overline{\mathbb{P}}$. 

In this article, the constant $B \left(>S_0\right)$ is an up-barrier fixed in the terms of the contract. The first hitting time of $B$ by the process $S_t$ is a random variable noted $\tau_{B}^{S}$ and defined as
\begin{equation}
\tau_{B}^{S}\equiv \inf \{t\geqslant 0, S_t\geqslant B\}.
\label{fpt}
\end{equation}
We also use the following notations :\\
$\bullet$ $\mu=\overline{\mu}/{\sigma}$ and $b=\log{\left(B/S_0\right)}/\sigma$.\\
$\bullet$ $\phi(x)$ and $\Phi(x)$ are the Gaussian density and distribution functions, respectively.\\
$\bullet$ $x_+=\max(x,0)$ and $ \mathbf{1}_{\{ \}}$ is an indicator function.

The payoff of GIC at maturity $T$ can be written as $\left(G_T-K\right)_+$, where $K$ is the strike price and $G_T$ is a random variable defined as
\begin{equation}
G_T\equiv\exp\left(  \frac{1}{T-\tau_{B}^{S}}
\displaystyle{\int_{\tau_{B}^{S}}^{T}}\log{S_u}\textrm{d}u\right) \mathbf{1}_{\{\tau_{B}^{S}<T\}}
+ S_T \mathbf{1}_{\{\tau_{B}^{S}\geqslant T\}}.
\label{GG}
\end{equation}
From definitions \eqref{fpt} and \eqref{GG}, we can see that the price of geometric Istanbul and geometric Asian call options coincide when $S_0\geqslant B$. Note that the geometric Istanbul put option whose a payoff at $T$ equal to $\left(K-G_T\right)_+$ is not priced here. As we will see, our analytical approximation method can be perfectly applied in the case of a put option.
\section{Pricing of geometric Istanbul options}\label{section3}
According to the risk-neutral pricing formula in continuous time\footnote{For more theoretical details and discussions on the general formula for pricing of derivative products under the assumptions of the continuous Black and Scholes model, see \citet{HARRISON1981215}.}, the price (or $premium$) at time 0 of GIO corresponds to expected value of its discounted payoff at maturity, this price will be noted for call option by $GIC_B$. Thus, we have
\begin{equation}
GIC_B=\mathbb{E}^{\overline{\mathbb{P}}}\left[e^{-r T}\left(G_T-K\right)_+\right],
\label{QQ1}
\end{equation}
where $\mathbb{E}^{\overline{\mathbb{P}}}$ is expectation operator under $\overline{\mathbb{P}}$-measure.  

The probability distribution of $G_T$ is essential in order to obtain an analytical formula of $GIC_B$. We notice that this distribution is known when $B$ is not reached before $T$. In this case, the distribution corresponds to the joint distribution of the geometric Brownian motion and its first hitting time of $B$. So, only the distribution when $B$ is reached before $T$ is unknown and need to be calculated.

For $x>0$, we have
\begin{align}
\overline{\mathbb{P}} \left(G_T\leqslant x, \tau_{B}^{S}<T \right)&=\displaystyle{\int_{0}^{T}}\overline{\mathbb{P}} \left(G_T\leqslant x, \tau_{B}^{S}=t \right)\textrm{d}t  \notag \\
&=\displaystyle{\int_{0}^{T}}\overline{\mathbb{P}} \left(\exp\left(  \frac{1}{T-\tau_{B}^{S}}
\displaystyle{\int_{\tau_{B}^{S}}^{T}}\log{S_u}\textrm{d}u\right) \leqslant x, \tau_{B}^{S}=t \right)\textrm{d}t. 
\end{align}
Let us introduce a process $Z_t$, $t \in [0, T]$, defined by $Z_t=W_{\tau_{B}^{S}+t}-W_{\tau_{B}^{S}}$. According to the strong Markov property on event $\{\tau_{B}^{S}<T\}$, the process $Z_t$ is a standard Brownian motion under $\overline{\mathbb{P}}$-measure. This process is started at zero and completely independent of stopping time $\tau_{B}^{S}$. 

Now we can write 
\begin{align}
\overline{\mathbb{P}} \left(\exp\left(  \frac{1}{T-\tau_{B}^{S}}
\displaystyle{\int_{\tau_{B}^{S}}^{T}}\log{S_u}\textrm{d}u\right) \leqslant x, \tau_{B}^{S}=t \right) \notag\\ 
&\!\!\!\!\!\!\!\!\!\!\!\!\!\!\!=\overline{\mathbb{P}} \left( \frac{\overline{\mu}}{2\sigma}\left(T-\tau_{B}^{S}\right)+ \frac{1}{T-\tau_{B}^{S}}\displaystyle{\int_{0}^{T-\tau_{B}^{S}}}Z_u\textrm{d}u \leqslant \frac{1}{\sigma}\log{\left(\frac{x}{B}\right) }, \tau_{B}^{S}=t \right) \notag  \\&\!\!\!\!\!\!\!\!\!\!\!\!\!\!\!
=\overline{\mathbb{P}} \left(  \frac{1}{T-t}\displaystyle{\int_{0}^{T-t}}Z_u\textrm{d}u \leqslant \frac{1}{\sigma}\log{\left(\frac{x}{B}\right) -\frac{\overline{\mu}}{2\sigma}\left(T-t\right)} \right) h(t)  \notag\\&\!\!\!\!\!\!\!\!\!\!\!\!\!\!\!
=\Phi \left( \frac{\sqrt{3}\left( \log(x/B) -\frac{\overline{\mu}}{2}(T-t)\right) }{\sigma\sqrt{T-t}}\right) h(t),
\label{Q8}
\end{align}
where 
\begin{equation}
h(t)=\frac{b}{\sqrt{2 \pi t^3}}\exp\left({-\frac{\left(b-\mu t\right)^2}{2t}}\right),
\label{FPTD}
\end{equation}
is the probability density function of first hitting time of $B$ by the process  $S_t$ (see formula 2.0.2 in \citet{borodin2002handbook}). Note that the equality \eqref{Q8} follows from the fact that the random variable $\int_{0}^{T-t}Z_u\textrm{d}u$ is Gaussian for $0 \leqslant t<T$ with a zero mean and a variance equal to $(T-t)^3/3$ (see \citet{Zhang}).

Then, the distribution of $G_T$ when $B$ is reached before $T$,  is given by the following formula 
\begin{equation}
\overline{\mathbb{P}} \left(G_T\leqslant x, \tau_{B}^{S}<T \right)=\displaystyle{\int_{0}^{T}}\Phi \left( \frac{\sqrt{3}\left( \log(x/B) -\frac{\overline{\mu}}{2}(T-t)\right) }{\sigma\sqrt{T-t}}\right) h(t)\textrm{d}t.  
\label{QQ3}
\end{equation}
The derivative of formula \eqref{QQ3} with respect to $x$ gives
\begin{align}
\begin{split}
\overline{\mathbb{P}} \left(G_T \in  (x, x+\textrm{d}x), \tau_{B}^{S}<T \right) &=
\frac{\sqrt{3} b}{2\pi \sigma x}\exp \left( \frac{3\mu\log\left( x/B\right)}{2\sigma} -\frac{3\mu^2 T}{8}+b\mu \right)\\
 & \quad \times \displaystyle{\int_{0}^{T} }\frac{1}{\sqrt{(T-t)t^3}}\exp   \left( -\frac{3\log^2\left( x/B\right)}{2(T-t)\sigma^2} -\frac{\mu^2}{8}t-\frac{b^2}{2t}\right)  \textrm{d}t \textrm{d}x,
\end{split}
\label{QQ4}
\end{align}
where $\textrm{d}x$ is an infinitesimal quantity.
\begin{Remark}
The integral in \eqref{QQ4} does not admit a closed-form expression; however, it is possible to approach it numerically with Gaussian quadrature methods (see \citet{brass2011quadrature}). As illustrated in \reffig{figure1}, the quantity $\mu^2/8$ is very small for a wide range of parameters $r$ and $\sigma$. This observation will allow us to obtain an analytical approximation of \eqref{QQ4} using a Taylor series expansion around zero.
\end{Remark}
\begin{figure}[H]
 \begin{center}
  \includegraphics[width=9.5cm, height=4cm]{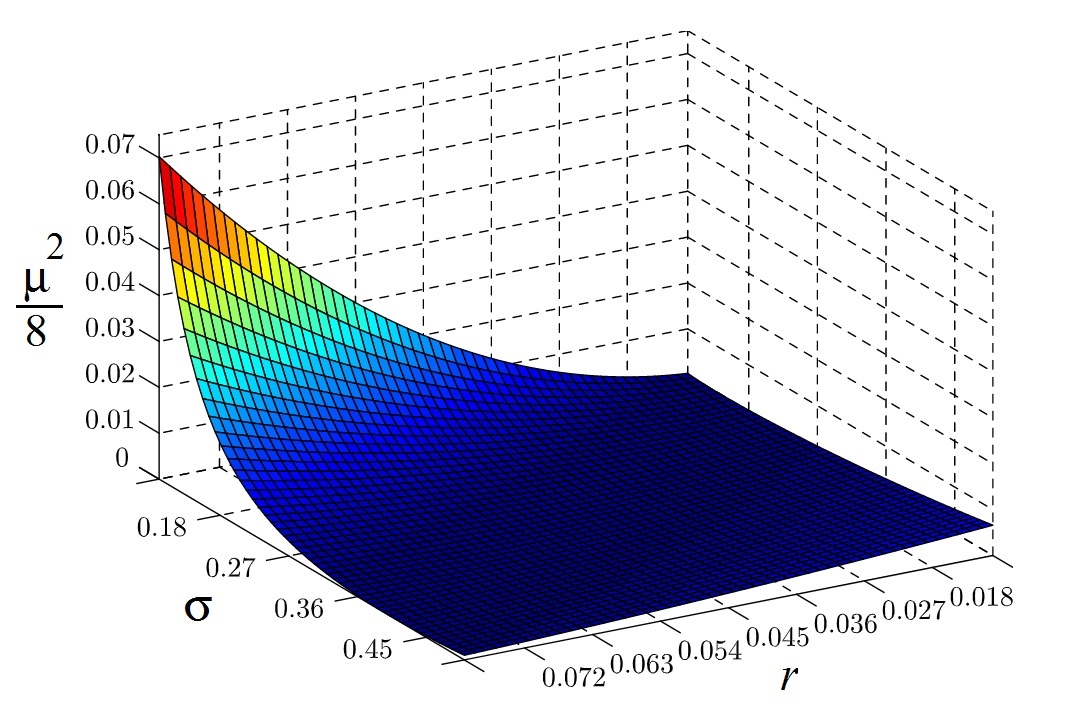}\\
  \caption{Values of $\mu^2/8$ for $r$ from $1\%$ to $8\%$ and $\sigma$ from $10\%$ to $50\%$.}\label{figure1}
 \end{center}
\end{figure}
\begin{Lemma}\label{lem1}
For $\alpha\geqslant 0 $, $\gamma$ and $T>0$, if $\beta$ is around zero, then we have
\begin{align*}
\frac{1}{\pi} \displaystyle{\int_{0}^{T} }\frac{1}{\sqrt{(T-t)t^3}}\exp   \left( -\frac{\alpha^2}{2(T-t)} -\beta t-\frac{\gamma}{t}\right)  \textrm{d}t
 &= \left( \frac{\alpha^2-2\gamma+T}{2}\beta^2-2\beta\right)\left( 1-\Phi(d)\right)\\&\quad
+\left(\frac{\sqrt{T}(\sqrt{2\gamma}-\alpha)}{2}\beta^2 +\sqrt{\frac{2}{T\gamma }}\right)\phi(d) +O(\beta^3),
\end{align*}
where $d={(\alpha+\sqrt{2\gamma})}/{\sqrt{T}}$.
\end{Lemma}
\begin{proof}
See \ref{AppendixC}.
\end{proof}
\begin{Theorem}\label{the1}
Suppose that $K\geqslant B$. We have
\begin{align}
GIC_B &\approx \frac{\sqrt{3}b}{2\sigma}\exp\left( -\frac{3\mu^2 T}{8}+b\mu-r T\right)\notag\\&\quad\times \Bigg\{B\bigg\{ \exp(z_3)\Big\{ z_4(1-\Phi(z_2))+z_6\phi(z_2)+z_7(1-\Phi(z_2))\Big\}+z_5(1-\Phi(z_1))\bigg\}\notag \\&\quad-K\bigg\{ \exp(z_9)\Big\{ z_{10}(1-\Phi(z_8))+\left( z_{12}+{w}/{a^2}\right)\phi(z_8)+z_{13}(1-\Phi(z_8))\Big\}+z_{11}(1-\Phi(z_1))\bigg\}   \Bigg\},
\label{T1}
\end{align}
where\\ 
\begin{minipage}{0.25\linewidth}
$a=\frac{\sqrt{3}}{\sigma\sqrt{T}}$,\\$h=\frac{|b|}{\sqrt{T}}$,
\end{minipage}
\begin{minipage}{0.25\linewidth}
$c=\frac{3\mu}{2\sigma}+1$,\\$k=\frac{(T-b^2)\mu^4}{128}-\frac{\mu^2}{4}$,
\end{minipage}
\begin{minipage}{0.25\linewidth}
$ d=\frac{3\mu^4}{128\sigma^2}$,\\$ l=\frac{2}{Th}+\frac{T\mu^4h}{128}$,
\end{minipage}
\begin{minipage}{0.50\linewidth}
 $e=c-1$,\\$w=-\frac{\mu^4\sqrt{3T}}{128\sigma}$,
\end{minipage}
\\
\begin{minipage}{0.50\linewidth}
$z_1=a\log{\left(\frac{K}{B}\right)}+h$,\\
$z_3=\frac{c^2}{2a^2}-\frac{hc}{a}$,\\ 
$z_5=\left(\frac{K}{B}\right)^c\left(-\frac{d\log^2(K/B)}{c}+\frac{2d\log(K/B)}{c^2}-\frac{2d}{c^3} -\frac{k}{c}\right)$,\\
$z_7=\frac{wc}{a^3}-\frac{wh}{a^2}+\frac{l}{a}$,\\
$ z_9=\frac{e^2}{2a^2}-\frac{he}{a}$,\\
$ z_{11}=\left(\frac{K}{B}\right)^e\left(-\frac{d\log^2(K/B)}{e}+\frac{2d\log(K/B)}{e^2}-\frac{2d}{e^3} -\frac{k}{e}\right)$,\\
$ z_{13}=\frac{we}{a^3}-\frac{wh}{a^2}+\frac{l}{a}.$
\end{minipage}
 \begin{minipage}{0.50\linewidth}
$z_2=z_1-\frac{c}{a}$,\\
$z_4=-\frac{2dh}{a^3}-\frac{d(1-h^2)}{ca^2}+\frac{2d}{c^3}+\frac{2dh}{ac^2}+\frac{dc}{a^4}+\frac{k}{c}$,\\
$z_6=\frac{d\log(K/B)}{ac}-\frac{2d}{ac^2}-\frac{dh}{ca^2}+\frac{d}{a^3}+\frac{w}{a^2}$,\\
$z_8=z_1-\frac{e}{a}$,\\
$z_{10}=-\frac{2dh}{a^3}-\frac{d(1-h^2)}{ea^2}+\frac{2d}{e^3}+\frac{2dh}{ae^2}+\frac{de}{a^4}+\frac{k}{e}$,\\
$z_{12}=\frac{d\log(K/B)}{ae}-\frac{2d}{ae^2}-\frac{dh}{ea^2}+\frac{d}{a^3}$,\\ \end{minipage}

\end{Theorem}
\begin{pot1}
We start by rewriting formula \eqref{QQ1} as 
\begin{equation}
GIC_B=\mathbb{E}^{\overline{\mathbb{P}}}\left[e^{-r T}\left(G_T-K\right)_+\mathbf{1}_{\{\tau_{B}^{S}<T\}}\right]+UOC_B,
\label{QQ13}
\end{equation}
where  $UOC_B$ is a price of an up-and-out barrier call option at time 0. The first term in  \eqref{QQ13} is written as follows
\begin{equation}
\mathbb{E}^{\overline{\mathbb{P}}}\left[e^{-r T}\left(G_T-K\right)_+\mathbf{1}_{\{\tau_{B}^{S}<T\}}\right]=
\displaystyle{\int_{K}^{+\infty} }e^{-rT}\left(x-K\right)\overline{\mathbb{P}} \left(G_T \in  (x, x+\textrm{d}x), \tau_{B}^{S}<T \right).
\label{QQ14}
\end{equation}
Using \reflem{lem1}, we obtain the following approximation 
\begin{equation*}
\mathbb{E}^{\overline{\mathbb{P}}}\left[e^{-r T}\left(G_T-K\right)_+\mathbf{1}_{\{\tau_{B}^{S}<T\}}\right]\approx  e^{-rT}\left( \mathcal{A}-K\mathcal{B}\right),
\end{equation*}
where
\begin{align*}
 \mathcal{A}&=\frac{B b\sqrt{3}}{2 \sigma }\exp \left(-\frac{3\mu^2 T}{8}+b\mu \right)\\&\quad \times
\Bigg\{ \displaystyle{\int_{\log \left(\frac{K}{B}\right) }^{+\infty}e^{(\frac{3\mu}{2\sigma}+1)z} \left( \left(\frac{3}{\sigma^2}z^2+T-b^2 \right)\frac{\mu^4}{128}-\frac{\mu^2}{4}\right)
\left(1-\Phi\left(\frac{\sqrt{3}}{\sigma\sqrt{T}}|z|+\frac{|b|}{\sqrt{T}}\right) \right)\textrm{d}z }\\&\quad + \displaystyle{\int_{\log \left(\frac{K}{B}\right) }^{+\infty}e^{(\frac{3\mu}{2\sigma}+1)z}\left(\left(|b|-\frac{\sqrt{3}}{\sigma}|z|\right)\frac{\mu^4\sqrt{T}}{128}+\frac{2}{|b|\sqrt{T}}\right) \phi\left(\frac{\sqrt{3}}{\sigma\sqrt{T}}|z|+\frac{|b|}{\sqrt{T}}\right)\textrm{d}z}\Bigg\}
\end{align*}
and
\begin{align*}
 \mathcal{B}&=\frac{b\sqrt{3}}{2 \sigma }\exp \left(-\frac{3\mu^2 T}{8}+b\mu \right)\\&\quad \times
\Bigg\{ \displaystyle{\int_{\log \left(\frac{K}{B}\right) }^{+\infty}e^{\frac{3\mu}{2\sigma}z} \left( \left(\frac{3}{\sigma^2}z^2+T-b^2 \right)\frac{\mu^4}{128}-\frac{\mu^2}{4}\right)
\left(1-\Phi\left(\frac{\sqrt{3}}{\sigma\sqrt{T}}|z|+\frac{|b|}{\sqrt{T}}\right) \right)\textrm{d}z }\\&\quad + \displaystyle{\int_{\log \left(\frac{K}{B}\right) }^{+\infty}e^{\frac{3\mu}{2\sigma}z}\left(\left(|b|-\frac{\sqrt{3}}{\sigma}|z|\right)\frac{\mu^4\sqrt{T}}{128}+\frac{2}{|b|\sqrt{T}}\right)\phi\left(\frac{\sqrt{3}}{\sigma\sqrt{T}}|z|+\frac{|b|}{\sqrt{T}}\right)\textrm{d}z}\Bigg\}.
\end{align*}
Since $K\geqslant B$, then $UOC_B$ has a zero value. It is sufficient to calculate the quantities  $\mathcal{A}$ and $\mathcal{B}$ with formulas \eqref{A3.} and \eqref{A4.} to have the desired result.
\end{pot1}
\begin{Theorem}\label{the2}
Suppose that $K< B$. We have
\begin{align}
GIC_B &\approx \frac{\sqrt{3}b}{2\sigma}\exp\left( -\frac{3\mu^2 T}{8}+b\mu-r T\right)\notag\\&\quad \times \Bigg\{B\bigg\{ \exp(z_3)\Big\{ z_4(\Phi(z_2)-\Phi(z_1))-z_5\phi(z_2)+\left(d\log(B/K)/(ca)+z_5\right)\phi(z_1) +\exp(-2hc/a)\notag\\&\quad\times\big\{z_6(\Phi(z_2-2c/a)-1)-(z_5+2dh/(ca^2))\phi(z_2-2c/a)\big\}\Big\}+z_7(1-\Phi(z_1-c/a))(K/B)^c \bigg\} 
\notag\\&\quad-K\bigg\{\exp(z_{10})\Big\{ z_{11}(\Phi(z_9)-\Phi(z_8))-z_{12}\phi(z_9)+\left(d\log(B/K)/(ea)+z_{12}\right)\phi(z_8)+\exp(-2he/a)\notag\\&\quad\times\big\{z_{13}(\Phi(z_9-2e/a)-1)-(z_{12}+2dh/(ea^2))\phi(z_9-2e/a)\big\} \Big\}+z_{14}(1-\Phi(z_8-e/a))(K/B)^e\bigg\}\Bigg\}\notag\\&\quad+UOC_B,
\label{T2}
\end{align}
where\\
 \begin{minipage}{0.25\linewidth}
$a=\frac{\sqrt{3}}{\sigma\sqrt{T}}$,\\$h=\frac{|b|}{\sqrt{T}}$,
\end{minipage}
\begin{minipage}{0.25\linewidth}
$c=\frac{3\mu}{2\sigma}+1$,\\$k=\frac{(T-b^2)\mu^4}{128}-\frac{\mu^2}{4}$,
\end{minipage}
\begin{minipage}{0.25\linewidth}
$d=\frac{3\mu^4}{128\sigma^2}$,\\$l=\frac{2}{Th}+\frac{T\mu^4h}{128}$,
\end{minipage}
\begin{minipage}{0.25\linewidth}
 $e=c-1$,\\$w=-\frac{\mu^4\sqrt{3T}}{128\sigma}$,
\end{minipage}
\\
\begin{minipage}{0.50\linewidth}
$z_1=a\log\left(\frac{B}{K}\right)+h+\frac{c}{a}$,\\
$z_3=\frac{c^2}{2a^2}+\frac{hc}{a}$,\\ 
$z_5=\frac{2d}{ac^2}-\frac{dh}{ca^2}-\frac{d}{a^3}-\frac{w}{a^2}$,\\
$z_7=-\frac{d}{c}\log^2\left(\frac{K}{B}\right)+\frac{2d}{c^2}\log\left(\frac{K}{B}\right)-\frac{2d}{c^3}-\frac{k}{c}$,\\
$z_9=z_8-a\log\left(\frac{B}{K}\right)$,\\
$z_{11}=\frac{2dh}{a^3}-\frac{d(1-h^2)}{ea^2}+\frac{2d}{e^3}-\frac{2dh}{ae^2}+\frac{de}{a^4}+\frac{k}{e}+\frac{we}{a^3}+\frac{wh}{a^2}-\frac{l}{a}$,\\
$ z_{13}=2\left(\frac{2hd}{a^3}-\frac{2hd}{e^2a}+\frac{wh}{a^2}-\frac{l}{a}\right)-z_{11}$,
\end{minipage}
 \begin{minipage}{0.50\linewidth}
$z_2=z_1-a\log\left(\frac{B}{K}\right)$,\\
$z_4=\frac{2dh}{a^3}-\frac{d(1-h^2)}{ca^2}+\frac{2d}{c^3}-\frac{2dh}{ac^2}+\frac{dc}{a^4}+\frac{k}{c}+\frac{wc}{a^3}+\frac{wh}{a^2}-\frac{l}{a}$,\\
$z_6=2\left(\frac{2hd}{a^3}-\frac{2hd}{c^2a}+\frac{wh}{a^2}-\frac{l}{a}\right)-z_4$,\\
$z_8=a\log\left(\frac{B}{K}\right)+h+\frac{e}{a}$,\\
$z_{10}=\frac{e^2}{2a^2}+\frac{he}{a}$,\\
$z_{12}=\frac{2d}{ae^2}-\frac{dh}{ea^2}-\frac{d}{a^3}-\frac{w}{a^2}$,\\
$z_{14}=-\frac{d}{e}\log^2\left(\frac{K}{B}\right)+\frac{2d}{e^2}\log\left(\frac{K}{B}\right)-\frac{2d}{e^3}-\frac{k}{e}$
 \end{minipage}
\\
\\
and $UOC_B$ is a price of an up-and-out barrier call option at time 0.
\end{Theorem}
\begin{pot2}
The proof is similar to that of Theorem 1, it should just be noted that the price $UOC_B$ is nonzero when $K<B$. Its value  for $S_0\leqslant B$ is known in closed-form (see formula (7.3.19) in \citet{shreve2004stochastic}).
\end{pot2}
\section{Numerical analysis}\label{section4}
In this section, we compare our analytical approximation formulas \eqref{T1} and \eqref{T2} with MC simulations. In our simulation procedure, we use the CV as a variance reduction technique of estimator obtained by crude MC method. We analyze two types of simulations errors, namely, the standard error and the relative error noted by S.E. and R.E., respectively. Our calculation algorithms are implemented with $\textsf{R}$ software version 3.5.1 on a PC, Dell, Intel(R) core(TM) i3, 1.70GHZ and running under Windows 8. To simulate the price \eqref{QQ1}, we start by discretizing the interval $[0,T]$ into $n=2500$ points, $0=t_0<t_1<...<t_n=T$, with the discretization step $\Delta t=T/n$. The simulation of model \eqref{GBM} is given by the following recursion formula  
\begin{equation}
S_{t_{i+1}}=S_{t_{i}}\exp \left( \overline{\mu} \Delta t + \sigma \sqrt{\Delta t}Y_{i+1} \right), \ \     i \in \big\{0,1, ... ,n-1 \big\} ,
\end{equation}
where $Y_1, Y_2, ... ,Y_n$ is $n$ i.i.d. standard Gaussian random variables. In order to obtain a realization of the random variable $G_T$, the time integral in \eqref{GG} is approximated with trapezoidal rule as follows
\begin{equation}
\displaystyle{\int_{\tau_{B}^{S}}^{T}}\log{S_u}\textrm{d}u \approx \sum_{i=t_B^S/\Delta t}^{n-1}\frac{\log{\left( S_{t_{i}}S_{t_{i+1}}\right)}}{2}\Delta t,
\end{equation}
where  $t_B^S= \inf \big\{t_i, i \in \{0,1, ... ,n-1 \} | S_{t_{i}}\geqslant B\big\}$ is discret version of first hitting time $\tau_B^S$. The number of paths used in our MC simulations is 10000. We take as a control variate the payoff of a geometric Asian call option (GAC) since the payoff of this option depends on  $S_{t_{0}}, S_{t_{1}}, ... , S_{t_{n}}$, which gives a high correlation with the payoff of our option. Our controlled estimator for $GIC_B$ is given by
\begin{equation}
 G\widehat{I}C_{B}^{CV}=G \widehat{I}C_{B}^{MC}-\theta^\star\left(G\widehat{A}C^{MC}-GAC\right),
\label{AV}
\end{equation}
where $G \widehat{I}C_{B}^{MC}$ and $G \widehat{A}C^{MC}$ are a crude MC estimators for $GIC_B$ and $GAC$ respectively and $\theta^\star$ is a parameter that minimizes the variance of  $G\widehat{I}C_{B}^{CV}$.\footnote{We take $\theta^\star={Cov(\mathcal{H}_1,\mathcal{H}_2)}/{Var(\mathcal{H}_2)}$, where $\mathcal{H}_1$ and $\mathcal{H}_2$ are the payoffs of geometric Istanbul and geometric Asian calls option, respectively. Note that the exact value of $\theta^\star$ is unknown we approximate it using the sample variance and covariance.} 

In \reftab{tab1}, we provide a comparison between the approximate price \eqref{T1} and the one obtained by MC simulations with the CV technique for different input parameters. The results obtained show that our approximation is efficient and could be applied in finance since the relative errors do not exceed $1.33\%$. The results in \reftab{tab1} also show that the option price increases as $K$ approaches $B$. Similarly, for \reftab{tab2}, the relative errors obtained with formula \eqref{T2} are all strictly less than $1.35\%$. This confirms once again that the price we provide for GIC is stable to changes in input parameters. We also observe from the results in \reftab{tab2} that the option price decreases as $K$ approaches $B$. It remains to be noted that in both Tables \ref{tab1} and \ref{tab2} the option price increases for longer expiration date, which is expected because the price of any type of option depends directly on its time-value.
\begin{table}[H]
\centering
\resizebox{0.86\textwidth}{!}{\begin{minipage}{\textwidth}
\caption{Comparison of geometric Istanbul call price obtained by our analytical approximation formula \eqref{T1} to that obtained by Monte-Carlo simulations.}
\label{tab1}
 \begin{tabular}{llllllllllllll}
 \hline
 &&&$T=0.5$&&&&$T=1$&&&&$T=1.5$&&\\
 \cline{4-6}\cline{8-10}\cline{12-14}
$S_0$&$K$&$B$&Approx.&MCV&R.E.&&Approx.&MCV&R.E.&&Approx.&MCV&R.E.\\
&&&&(S.E.)&(\%)&&&(S.E.)&(\%)&&&(S.E.)&(\%)\\
\hline
$57$&$63$&$60$&$1.2886$&$\;1.2828$&$0.4521$&&$2.4889$&$\;2.5103$&$0.8489$&&$3.4720$&$\;3.5082$&$1.0317$\\
&&&&$(0.0087)$&&&&$(0.0127)$&&&&$(0.0158)$&\\
$58$&$63$&$60$&$1.4739$&$\;1.4864$&$0.8415$&&$2.7201$&$\;2.7566$&$1.3225$&&$3.7257$&$\;3.7531$&$0.7307$\\
&&&&$(0.0081)$&&&&$(0.0112)$&&&&$(0.0132)$&\\
$59$&$63$&$60$&$1.6747$&$\;1.6863$&$0.6860$&&$2.9622$&$\;2.9982$&$1.1991$&&$3.9878$&$\;4.0251$&$0.9256$\\
&&&&$(0.0194)$&&&&$(0.0088)$&&&&$(0.0100)$&\\
\hline
$60$&$63$&$63$&$2.4187$&$\;2.4414$&$0.9283$&&$3.8050$&$\;3.8491$&$1.1477$&&$4.8783$&$\;4.9095$&$0.6362$\\
&&&&$(0.0122)$&&&&$(0.0160)$&&&&$(0.0178)$&\\
$60$&$64$&$63$&$2.0400$&$\;2.0585$&$0.9015$&&$3.4023$&$\;3.4367$&$1.0024$&&$4.4704$&$\;4.4746$&$0.0945$\\
&&&&$(0.0113)$&&&&$(0.0149)$&&&&$(0.0162)$&\\
$60$&$65$&$63$&$1.7079$&$\;1.7226$&$0.8529$&&$3.0328$&$\;3.0610$&$0.9210$&&$4.0893$&$\;4.1146$&$0.6148$\\
&&&&$(0.0105)$&&&&$(0.0141)$&&&&$(0.0161)$&\\
\hline
$70$&$75$&$72$&$2.0299$&$\;2.0501$&$0.9853$&&$3.5694$&$\;3.5975$&$0.7810$&&$4.7936$&$\;4.8402$&$0.9626$\\
&&&&$(0.0095)$&&&&$(0.0126)$&&&&$(0.0145)$&\\
$70$&$75$&$73$&$2.1844$&$\;2.1984$&$0.6390$&&$3.7503$&$\;3.7965$&$1.2176$&&$4.9874$&$\;5.0418$&$1.0803$\\
&&&&$(0.0116)$&&&&$(0.0158)$&&&&$(0.0187)$&\\
$70$&$75$&$75$&$2.5116$&$\;2.5353$&$0.9347$&&$4.1237$&$\;4.1585$&$0.8350$&&$5.3831$&$\;5.4315$&$0.8903$\\
&&&&$(0.0166)$&&&&$(0.0210)$&&&&$(0.0243)$&\\
\hline
\multicolumn{14}{@{}p{172mm}}{\footnotesize{Notes: The input parameters are taken as follows: $r=0.05$ and $\sigma=0.3$. We note by "Approx." the price of geometric Istanbul call option obtained with formula \eqref{T1} and by "MCV" the Monte-Carlo estimator of the price of the same option using the control variates method. We also note by"S.E." the standard error of MCV and by "R.E." the relative error wich is given in percentage with the following formula: $R.E.=\frac{|Approx.-MCV|}{MCV}\times 100\%$.}}
\end{tabular}
\end{minipage}}
\centering
\resizebox{0.86\textwidth}{!}{\begin{minipage}{\textwidth}
\caption{Comparison of geometric Istanbul call price obtained by our analytical approximation formula \eqref{T2} to that obtained by Monte-Carlo simulations. }
\label{tab2}
 \begin{tabular}{llllllllllllll}
 \hline
 &&&$T=0.5$&&&&$T=1$&&&&$T=1.5$&&\\
 \cline{4-6}\cline{8-10}\cline{12-14}
$S_0$&$K$&$B$&Approx.&MCV&R.E.&&Approx.&MCV&R.E.&&Approx.&MCV&R.E.\\
&&&&(S.E.)&(\%)&&&(S.E.)&(\%)&&&(S.E.)&(\%)\\
\hline
$55$&$56$&$58$&$3.0603$&$\;3.0859$&$0.8327$&&$4.3377$&$\;4.3858$&$1.0980$&&$5.3139$&$\;5.3606$&$0.8696$\\
&&&&$(0.0136)$&&&&$(0.0167)$&&&&$(0.0184)$&\\
$56$&$56$&$58$&$3.3988$&$\;3.4029$&$0.1205$&&$4.6770$&$\;4.7357$&$1.2408$&&$5.6544$&$\;5.7237$&$1.2112$\\
&&&&$(0.0113)$&&&&$(0.0143)$&&&&$(0.0161)$&\\
$57$&$56$&$58$&$3.7535$&$\;3.7905$&$0.9762$&&$5.0266$&$\;5.0444$&$0.3531$&&$6.0025$&$\;6.0535$&$0.8421$\\
&&&&$(0.0091)$&&&&$(0.0096)$&&&&$(0.0118)$&\\
\hline
$60$&$61$&$64$&$3.5470$&$\;3.5650$&$0.5039$&&$4.9452$&$\;4,9782$&$0,6634$&&$6.0113$&$\;6.0866$&$1.2374$\\
&&&&$(0.0165)$&&&&$(0.0201)$&&&&$(0.0233)$&\\
$60$&$62$&$64$&$3.0547$&$\;3.0800$&$0.8214$&&$4.4610$&$\;4.4886$&$0.6155$&&$5.5376$&$\;5.5609$&$0.4190$\\
&&&&$(0.0156)$&&&&$(0.0190)$&&&&$(0.0210)$&\\
$60$&$63$&$64$&$2.6087$&$\;2.6245$&$0.6041$&&$4.0103$&$\;4.0650$&$1.3445$&&$5.0911$&$\;5.1519$&$1.1792$\\
&&&&$(0.0146)$&&&&$(0.0188)$&&&&$(0.0211)$&\\
\hline
$79$&$81$&$82$&$3.8378$&$\;3.8695$&$0.8190$&&$5.6662$&$\;5.7112$&$0.7886$&&$7.0688$&$\;7.1464$&$1.0859$\\
&&&&$(0.0155)$&&&&$(0.0188)$&&&&$(0.0224)$&\\
$79$&$81$&$85$&$4.4841$&$\;4.4995$&$0.3419$&&$6.3405$&$\;6.4147$&$1.1562$&&$7.7554$&$\;7.8190$&$0.8137$\\
&&&&$(0.0227)$&&&&$(0.0282)$&&&&$(0.0308)$&\\
$79$&$81$&$87$&$4.9003$&$\;4.9394$&$0.7914$&&$6.7895$&$\;6.8440$&$0.7962$&&$8.2147$&$\;8.2225$&$0.0944$\\
&&&&$(0.0275)$&&&&$(0.0327)$&&&&$(0.0358)$&\\
\hline
\multicolumn{14}{@{}p{172mm}}{\footnotesize{Notes: The input parameters are taken as follows: $r=0.05$ and $\sigma=0.3$. We note by "Approx." the price of geometric Istanbul call option obtained with formula \eqref{T2} and by "MCV" the Monte-Carlo estimator of the price of the same option using the control variates method. We also note by"S.E." the standard error of MCV and by "R.E." the relative error wich is given in percentage with the following formula: $R.E.=\frac{|Approx.-MCV|}{MCV}\times 100\%$. }}
\end{tabular}
\end{minipage}}
\end{table}
In \reftab{tab3}, we analyze the robustness of approximation formulas \eqref{T1} and \eqref{T2} when the maturity date is long. Our analysis consists in adopting the same MC simulations strategy by increasing the maturity each time while fixing all the inputs. The results thus obtained show that the relative errors do not exceed $1.5\%$, which means that our analytical approximations remain both stable and efficient for long-term contracts.
\begin{table}[H]
\begin{center}
\caption{Relative errors when maturity is longer}
\label{tab3}
\begin{tabular}{llllll} 
\hline
Maturity $T$\,\,\,  & 2\,\,\,  & 3\,\,\, & 4\,\,\,  & 5\,\,\,  & 6\,\,\, \\ 
\hline
R.E. with approx. formula (\ref{T1})\,\,\, & 0.8260\%\,\,\,  & 1.2073\%\,\,\,  & 1.4955\%\,\,\,  & 0.4393\,\,\,   &1.0083\%\,\,\, \\
R.E. with approx. formula (\ref{T2}) & 0.4017\% & 0.9485\% & 1.2925\% & 0.7768 & 1.2080\%\\ 
\hline
\multicolumn{6}{@{}p{145mm}}{\footnotesize{Notes: The maturities are taken in years (first row), we consider contracts with a lifetime ranging from 2 to 6 years. For formula  \eqref{T1} (second row), the input parameters are: $r=0.05$, $\sigma=0.3$, $S_0=75$, $B=79$ and $K=80$.  For formula  \eqref{T2} (third row), the input parameters are: $r=0.05$, $\sigma=0.3$,  $S_0=55$, $B=58$ and $K=56$.} }
\end{tabular}
\end{center}
\end{table}
In \reffig{figure2}, we compare the price of the GIC to the AIC. As in \citet{MichelJac1}, we use the log-normal approximation method to estimate the price in the arithmetic case.  For the geometric case, we use our analytical approximation formulas \eqref{T1} and \eqref{T2}. The numerical results show that a GIC is relatively cheaper than a AIC. We also observe, on the left side of \reffig{figure2}, that the price of the Istanbul call option rises when the barrier is close to the current price for both types of averages. This observation is explained by the fact that the closer the barrier is to the current price, the higher the probability that it will be reached, thus increasing the theoretical value of the option. Furthermore, on the right side of \reffig{figure2}, we can see that the price of the Istanbul call option decreases as the strike price moves away from the current price, this is due to the fact that the probability of the option expire in-the-money becomes progressively lower as the strike price becomes higher than the current price.
\begin{figure}[H]
 \begin{center}
   \includegraphics[width=17cm,  height=6.5cm]{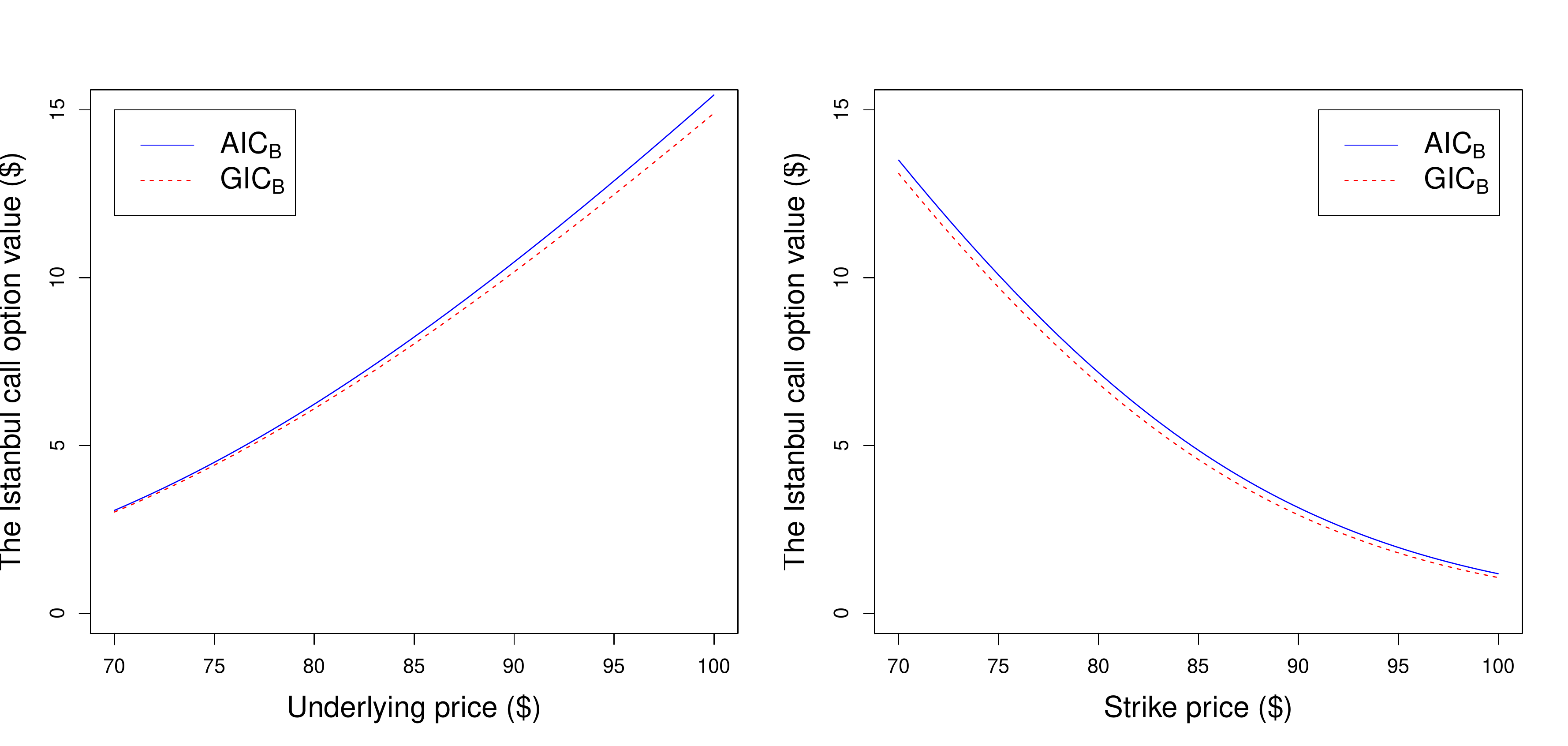}\\
  \caption{
Call price comparison between the geometric and arithmetic Istanbul options. Notes: The Istanbul call option prices are noted $GIC_B$ and $AIC_B$ for the geometric and arithmetic average cases, respectively. We take as input parameters for the left-hand plot: $S_0$ from 70 to 100, $\sigma=0.3$, $r=0.05$, $T=1$, $B=105$ and $K=90$. For the right-hand plot, we consider the input parameters: $K$ from $70$ to $100$, $\sigma=0.3$, $r=0.05$, $T=1$, $B=85$ and $S_0=79$. }\label{figure2}
 \end{center}
\end{figure}
We end this section with a study on the sensitivity of the price of an GIC to changes in the price of the underlying asset. For this purpose, we analyze an important risk measure which is the Delta ($\Delta$). This theoretical quantity is used by options traders to develop good investment strategies.\footnote{The Delta is considered by some traders as an approximation of the probability that the option will expire in-the-money.} In our case, the $\Delta$ of a geometric Istanbul call option corresponds to the partial derivative of \eqref{QQ1} with respect to $S_0$. In \reffig{figure3}, on the left side, we calculate the $\Delta$ values relative to the price of the underlying asset while increasing the volatility at each plot. On the right side, we fixed the underlying asset and calculate the $\Delta$ values relative to the strike price while increasing the maturity at each plot. As shown in Figure 3, the value of $\Delta$ depends on three main factors: moneyness, volatility, and maturity. It should be noted that $\Delta$ is constantly changing during the trading period and therefore does not predict the maturity value of the underlying asset price.
\begin{figure}[H]
 \begin{center}
   \includegraphics[width=17cm,  height=6.5cm]{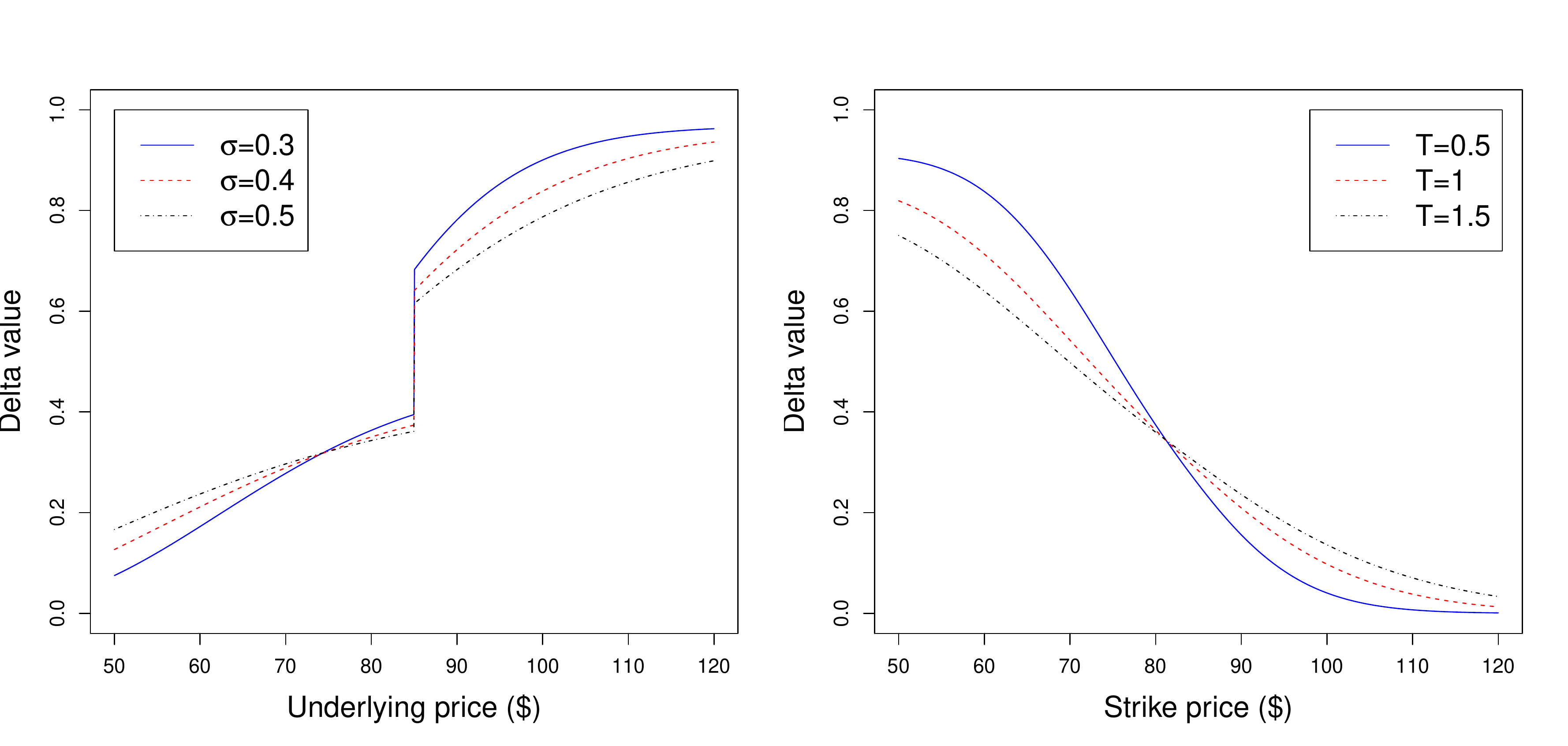}\\
  \caption{
The Delta of an geometric Istanbul call option. Notes: The left-hand plots are constructed with the following input parameters: $r=0.05$, $T=1$, $B=85$ and $K=80$. In the right-hand plots, we take: $r=0.05$, $\sigma=0.3$, $B=85$ and $S_0=80$. The value of $\Delta$ is approximated numerically by first-order finite difference.}\label{figure3}
 \end{center}
\end{figure}

\section{Conclusion}\label{section5}
In this paper, we addressed the pricing problem of geometric Istanbul options under the standard Black-Scholes model. A closed-form analytical approximation formula has been proposed for the price of a call option with a fixed strike price. The numerical results obtained by Monte-Carlo simulations using the control variates method have shown that our analytical approximation is very efficient for a wide range of input parameters and can therefore be used in finance.  In addition, we have shown through a comparative study that geometric Istanbul call options have a more attractive price compared to those with an arithmetic average treated by Michel Jacques in 1997. Finally, we illustrated, graphically, the price sensitivity of a geometric Istanbul call option to changes in the price of the underlying asset.

Future research on Istanbul options could follow two directions. The first would be to make changes to the input parameters, such as studying the case of a floating strike price, the adoption of a down-barrier or studying the case of a harmonic average.  The second interesting approach would be to extend the concept of Istanbul options to more complex economic models such as the exponential Levy model, the CEV model, the Heston model, etc.
\appendix
\section{Some formulas of indefinite integrals of Gaussian functions}
\begin{align}
\displaystyle{\int}\frac{1}{\sqrt{x^3}}\exp   \left( -ax -\frac{h}{2x}\right)  \textrm{d}x&=\sqrt{\frac{2\pi}{h}}\exp\left( -\sqrt{2ah}\right)\Phi\left( \sqrt{2ax}-\sqrt{\frac{h}{x}}\right)\notag\\&\quad
+\sqrt{\frac{2\pi}{h}}\exp\left( \sqrt{2ah}\right)\Phi\left(- \sqrt{2ax}-\sqrt{\frac{h}{x}}\right), 
\label{A1.}
\end{align}
where $a\geqslant0$ and $h> 0$.
\begin{align}
\begin{split}
\displaystyle{\int}\Phi\left(ax+h \right)\textrm{d}x=\left(x+\frac{h}{a}\right)\Phi\left(ax+h \right)+\frac{1}{a}\phi\left(ax+h \right),
\end{split}
\label{A2.}
\end{align} 
where $a\neq 0$.
\begin{align}
\displaystyle{\int}e^{cx}\left( dx^2+k \right)\left(1-\Phi\left(ax+h \right) \right)\textrm{d}x&=
\left(\frac{d}{c}x^2-\frac{2d}{c^2}x+\frac{2d}{c^3}+\frac{k}{c}\right)\left(1-\Phi\left(ax+h \right)\right)e^{cx}
-\exp\left( \frac{c^2}{2a^2}-\frac{hc}{a}\right)\notag\\&\quad \times\left(\frac{2hd}{a^3}+\frac{d(1-h^2)}{ca^2}-\frac{2d}{c^3}-\frac{2hd}{c^2 a}-\frac{dc}{a^4}-\frac{k}{c} \right) \Phi\left( ax+h-\frac{c}{a}\right)\notag\\&\quad
-\exp\left( \frac{c^2}{2a^2}-\frac{hc}{a}\right)\left(\frac{d}{ca}x-\frac{2d}{c^2 a}-\frac{hd}{c a^2}+\frac{d}{a^3}\right)\phi\left( ax+h-\frac{c}{a}\right), 
\label{A3.}
\end{align}
where $a,c\neq 0$.
\begin{align}
\displaystyle{\int}e^{cx}\left( wx+l \right)\phi\left( ax+h\right)\textrm{d}x=
\exp\left( \frac{c^2}{2a^2}-\frac{hc}{a}\right)\left(\left(\frac{wc}{a^3}-\frac{wh}{a^2}+\frac{l}{a}\right)\Phi\left( ax+h-\frac{c}{a}\right)-\frac{w}{a^2}\phi\left( ax+h-\frac{c}{a}\right)\right),
\label{A4.}
\end{align}
where $a\neq 0$.
\begin{proof}The formulas \eqref{A1.}, \eqref{A2.}, \eqref{A3.} and \eqref{A4.} can be proved by differentiation with respect to $x$.
\end{proof}
\section{A closed-form expressions of two integrals used in this article}
\begin{equation}
\displaystyle{\int_{0}^{T} }\frac{1}{\sqrt{(T-t)t}}\exp   \left( -\frac{\alpha^2}{2(T-t)}\right)  \textrm{d}t=2\pi\left(1-\Phi\left( \frac{\alpha}{\sqrt{T}}\right) \right),
\label{B.1}
\end{equation}
where $\alpha \geqslant 0$ and $T>0$.
\begin{equation}
\displaystyle{\int_{0}^{T} }\frac{t}{\sqrt{(T-t)t}}\exp   \left( -\frac{\alpha^2}{2t(T-t)}\right)  \textrm{d}t=T\pi \left(1-\Phi\left( \frac{2\alpha}{\sqrt{T}}\right)\right),
\label{B.2}
\end{equation}
where $\alpha \geqslant 0$ and $T>0$.
\begin{proof}
The formulas \eqref{B.1} and \eqref{B.2} can be easily found by using the convolution theorem and applying the table of Laplace transforms available in \citet{poularikas1998handbook}.
\end{proof}
\section{Proof of \reflem{lem1}}\label{AppendixC}
By a second-order Taylor series expansion around zero, we have
\begin{equation}
e^{-\beta t}=1-t \beta +\frac{t^2}{2}\beta^2+O(\beta^3),  
\label{QQ5}
\end{equation}
from equation \eqref{QQ5}, we have
\begin{equation}
\frac{1}{\pi} \displaystyle{\int_{0}^{T} }\frac{1}{\sqrt{(T-t)t^3}}\exp   \left( -\frac{\alpha^2}{2(T-t)} -\beta t-\frac{\gamma}{t}\right)  \textrm{d}t=A(\alpha,\gamma,T)- B(\alpha,\gamma,T)\beta+C(\alpha,\gamma,T)\frac{\beta^2}{2}+O(\beta^3),
\label{QQ1.2}
\end{equation}
where
\begin{align}
A(\alpha,\gamma,T)=\displaystyle{\int_{0}^{T} }\frac{1}{\sqrt{(T-t)t^3}}\exp   \left( -\frac{\alpha^2}{2(T-t)} -\frac{\gamma}{t}\right)\!\textrm{d}t,\!\!\!
\label{QQ6}\\ 
B(\alpha,\gamma,T)=\displaystyle{\int_{0}^{T} }\frac{1}{\sqrt{(T-t)t}}\exp   \left( -\frac{\alpha^2}{2(T-t)} -\frac{\gamma}{t}\right)\!\textrm{d}t,
\label{QQ7}\\
C(\alpha,\gamma,T)=\displaystyle{\int_{0}^{T} }\frac{t}{\sqrt{(T-t)t}}\exp   \left( -\frac{\alpha^2}{2(T-t)} -\frac{\gamma}{t}\right)\!\textrm{d}t.
\label{QQ8}
\end{align}
We start with the evaluation of \eqref{QQ6}, which can also be written as
\begin{equation}
 A(\alpha,\gamma,T)=\frac{1}{T}\exp\left( -\frac{\alpha^2+2\gamma}{2T} \right) \displaystyle{\int_{0}^{+\infty} }\frac{1}{\sqrt{x^3}}\exp   \left( -\frac{\alpha^2}{2T}x -\frac{\gamma}{Tx}\right)  \textrm{d}x,
\end{equation}
thanks to the formula \eqref{A1.}, we obtain  
\begin{equation*}
 A(\alpha,\gamma,T)=\pi\sqrt{\frac{2}{T\gamma}}\phi\left( \frac{\alpha+\sqrt{2\gamma}}{\sqrt{T}}\right). 
\end{equation*}
For formula \eqref{QQ7}, we start with the following relationship
\begin{equation}
\frac{\partial }{\partial \gamma}B(\alpha,\gamma,T)=-A(\alpha,\gamma,T),
\end{equation} 
by integration with respect to $\gamma$,  we obtain
\begin{equation}
B(\alpha,\gamma,T)=-2\pi\Phi\left( \frac{\alpha+\sqrt{2\gamma}}{\sqrt{T}}\right) +\varepsilon_{1}(\alpha,T),
\label{QQ9}
\end{equation}
where $\varepsilon_{1}(\alpha,T)$ is a function independent of $\gamma$.
The limit of \eqref{QQ9} when $\gamma$ tends to zero gives
\begin{equation}
\displaystyle{\int_{0}^{T} }\frac{1}{\sqrt{(T-t)t}}\exp   \left( -\frac{\alpha^2}{2(T-t)}\right)  \textrm{d}t=-2\pi\Phi\left( \frac{\alpha}{\sqrt{T}}\right) +\varepsilon_{1}(\alpha,T),
\label{QQ10}
\end{equation}
using the formula \eqref{B.1}, we get $\varepsilon_{1}(\alpha,T)=2\pi$ and therefore 
\begin{equation*}
B(\alpha,\gamma,T)=2\pi\left( 1-\Phi\left( \frac{\alpha+\sqrt{2\gamma}}{\sqrt{T}}\right) \right).  
\end{equation*}
Similarly, to obtain a closed-form expression of \eqref{QQ8}, we note that 
\begin{equation}
\frac{\partial }{\partial \gamma}C(\alpha,\gamma,T)=-B(\alpha,\gamma,T).
\end{equation}
The integration with respect to $\gamma$ taking $x=\sqrt{\gamma}$ and using formula \eqref{A2.} gives  
\begin{equation}
\begin{split}
C(\alpha,\gamma,T)=\pi \left( \left( 2\gamma-\alpha^2-T\right) \Phi \left( \frac{\alpha+\sqrt{2\gamma}}{\sqrt{T}}\right) + \sqrt{T}\left(\sqrt{2\gamma}-\alpha\right)  \phi\left( \frac{\alpha+\sqrt{2\gamma}}{\sqrt{T}}\right)\right) -2\pi\gamma +\varepsilon_{2}(\alpha,T),
\end{split}
\label{QQ11}
\end{equation}
where $\varepsilon_{2}(\alpha,T)$ is a function that depends only on  $\alpha$ and $T$.
The limit of \eqref{QQ11} when $\gamma$ tends to $\alpha^2/2$ gives
\begin{equation}
\displaystyle{\int_{0}^{T} }\frac{t}{\sqrt{(T-t)t}}\exp   \left( -\frac{\alpha^2 T}{2t(T-t)}\right)  \textrm{d}t=-T\pi\Phi\left( \frac{2\alpha}{\sqrt{T}}\right) -\alpha^2\pi+\varepsilon_{2}(\alpha,T).
\label{QQ12}
\end{equation}
Using the formula \eqref{B.2}, we get $\varepsilon_{2}(\alpha,T)=\pi\left( T+\alpha^2\right) $ and therefore 
\begin{equation*}
C(\alpha,\gamma,T)=\pi \left( \left( 2\gamma-\alpha^2-T\right) \Phi \left( \frac{\alpha+\sqrt{2\gamma}}{\sqrt{T}}\right) + \sqrt{T}\left(\sqrt{2\gamma}-\alpha\right)  \phi\left( \frac{\alpha+\sqrt{2\gamma}}{\sqrt{T}}\right)-2\gamma+T+\alpha^2\right).
\end{equation*}
Finally, it is sufficient to replace the formulas \eqref{QQ6}, \eqref{QQ7} and \eqref{QQ8} in \eqref{QQ1.2} to obtain the desired result.
\bibliographystyle{elsarticle-harv}
\bibliography{bibliog}

\begin{thebibliography}{21}
\expandafter\ifx\csname natexlab\endcsname\relax\def\natexlab#1{#1}\fi
\providecommand{\url}[1]{\texttt{#1}}
\providecommand{\href}[2]{#2}
\providecommand{\path}[1]{#1}
\providecommand{\DOIprefix}{doi:}
\providecommand{\ArXivprefix}{arXiv:}
\providecommand{\URLprefix}{URL: }
\providecommand{\Pubmedprefix}{pmid:}
\providecommand{\doi}[1]{\href{http://dx.doi.org/#1}{\path{#1}}}
\providecommand{\Pubmed}[1]{\href{pmid:#1}{\path{#1}}}
\providecommand{\bibinfo}[2]{#2}
\ifx\xfnm\relax \def\xfnm[#1]{\unskip,\space#1}\fi
\bibitem[{Aimi et~al.(2018)Aimi, Diazzi and Guardasoni}]{Aimi2}
\bibinfo{author}{Aimi, A.}, \bibinfo{author}{Diazzi, L.},
  \bibinfo{author}{Guardasoni, C.}, \bibinfo{year}{2018}.
\newblock \bibinfo{title}{Efficient bem-based algorithm for pricing floating
  strike asian barrier options (with matlab ® code)}.
\newblock \bibinfo{journal}{Axioms} \bibinfo{volume}{7}, \bibinfo{pages}{40}.
\bibitem[{Aimi and Guardasoni(2017)}]{Aimi3}
\bibinfo{author}{Aimi, A.}, \bibinfo{author}{Guardasoni, C.},
  \bibinfo{year}{2017}.
\newblock \bibinfo{title}{Collocation boundary element method for the pricing
  of geometric asian options}.
\newblock \bibinfo{journal}{Engineering Analysis with Boundary Elements}
  \bibinfo{volume}{92}.
\bibitem[{Angus(1999)}]{Angus}
\bibinfo{author}{Angus, J.E.}, \bibinfo{year}{1999}.
\newblock \bibinfo{title}{A note on pricing asian derivatives with continuous
  geometric averaging}.
\newblock \bibinfo{journal}{Journal of Futures Markets} \bibinfo{volume}{19},
  \bibinfo{pages}{845--858}.
\bibitem[{Black and Scholes(1973)}]{BlackScholes11}
\bibinfo{author}{Black, F.}, \bibinfo{author}{Scholes, M.},
  \bibinfo{year}{1973}.
\newblock \bibinfo{title}{The pricing of options and corporate liabilities}.
\newblock \bibinfo{journal}{Journal of Political Economy} \bibinfo{volume}{81},
  \bibinfo{pages}{637--654}.
\bibitem[{Borodin and Salminen(2002)}]{borodin2002handbook}
\bibinfo{author}{Borodin, A.N.}, \bibinfo{author}{Salminen, P.},
  \bibinfo{year}{2002}.
\newblock \bibinfo{title}{Handbook of Brownian motion: facts and formulae}.
\newblock \bibinfo{publisher}{Springer}.
\bibitem[{Brass and Petras(2011)}]{brass2011quadrature}
\bibinfo{author}{Brass, H.}, \bibinfo{author}{Petras, K.},
  \bibinfo{year}{2011}.
\newblock \bibinfo{title}{Quadrature Theory: The Theory of Numerical
  Integration on a Compact Interval}.
\newblock Mathematical surveys and monographs, \bibinfo{publisher}{American
  Mathematical Society}.
\bibitem[{Forsyth and Vetzal(1999)}]{KK}
\bibinfo{author}{Forsyth, P.}, \bibinfo{author}{Vetzal, K.},
  \bibinfo{year}{1999}.
\newblock \bibinfo{title}{Discrete asian barrier options}.
\newblock \bibinfo{journal}{Journal of Computational Finance}
  \bibinfo{volume}{3}.
\bibitem[{Harrison and Pliska(1981)}]{HARRISON1981215}
\bibinfo{author}{Harrison, J.}, \bibinfo{author}{Pliska, S.R.},
  \bibinfo{year}{1981}.
\newblock \bibinfo{title}{Martingales and stochastic integrals in the theory of
  continuous trading}.
\newblock \bibinfo{journal}{Stochastic Processes and their Applications}
  \bibinfo{volume}{11}, \bibinfo{pages}{215 -- 260}.
\bibitem[{{Hsu} et~al.(2012){Hsu}, {Lu}, {Kao}, {Lyuu} and {Ho}}]{6327776}
\bibinfo{author}{{Hsu}, W.W.Y.}, \bibinfo{author}{{Lu}, C.},
  \bibinfo{author}{{Kao}, M.}, \bibinfo{author}{{Lyuu}, Y.},
  \bibinfo{author}{{Ho}, J.}, \bibinfo{year}{2012}.
\newblock \bibinfo{title}{Pricing discrete asian barrier options on lattices} ,
  \bibinfo{pages}{1--8}.
\bibitem[{Hubalek and Sgarra(2011)}]{HUBALEK20113355}
\bibinfo{author}{Hubalek, F.}, \bibinfo{author}{Sgarra, C.},
  \bibinfo{year}{2011}.
\newblock \bibinfo{title}{On the explicit evaluation of the geometric asian
  options in stochastic volatility models with jumps}.
\newblock \bibinfo{journal}{Journal of Computational and Applied Mathematics}
  \bibinfo{volume}{235}, \bibinfo{pages}{3355 -- 3365}.
\bibitem[{Jacques(1997)}]{MichelJac1}
\bibinfo{author}{Jacques, M.}, \bibinfo{year}{1997}.
\newblock \bibinfo{title}{The istanbul option: where the standard european
  option becomes asian.}
\newblock \bibinfo{journal}{Insurance: Mathematics and Economics}
  \bibinfo{volume}{21}, \bibinfo{pages}{139--152}.
\bibitem[{Ju(2014)}]{JU}
\bibinfo{author}{Ju, N.}, \bibinfo{year}{2014}.
\newblock \bibinfo{title}{Pricing asian and basket options via taylor
  expansion}.
\newblock \bibinfo{journal}{J. Comput. Finance} \bibinfo{volume}{5}.
\bibitem[{Kemna and Vorst(1990)}]{Kemna}
\bibinfo{author}{Kemna, A.G.Z.}, \bibinfo{author}{Vorst, A.C.F.},
  \bibinfo{year}{1990}.
\newblock \bibinfo{title}{{A pricing method for options based on average asset
  values}}.
\newblock \bibinfo{journal}{Journal of Banking \& Finance}
  \bibinfo{volume}{14}, \bibinfo{pages}{113--129}.
\bibitem[{Levy(1992)}]{LEVY1992474}
\bibinfo{author}{Levy, E.}, \bibinfo{year}{1992}.
\newblock \bibinfo{title}{Pricing european average rate currency options}.
\newblock \bibinfo{journal}{Journal of International Money and Finance}
  \bibinfo{volume}{11}, \bibinfo{pages}{474 -- 491}.
\bibitem[{Lu et~al.(2019)Lu, Liang, Hsieh and Lee}]{LU2019}
\bibinfo{author}{Lu, K.J.}, \bibinfo{author}{Liang, C.J.},
  \bibinfo{author}{Hsieh, M.H.}, \bibinfo{author}{Lee, Y.H.},
  \bibinfo{year}{2019}.
\newblock \bibinfo{title}{An effective hybrid variance reduction method for
  pricing the asian options and its variants}.
\newblock \bibinfo{journal}{The North American Journal of Economics and
  Finance} .
\bibitem[{Mehrdoust(2015)}]{CC}
\bibinfo{author}{Mehrdoust, F.}, \bibinfo{year}{2015}.
\newblock \bibinfo{title}{A new hybrid monte carlo simulation for asian options
  pricing}.
\newblock \bibinfo{journal}{Journal of Statistical Computation and Simulation}
  \bibinfo{volume}{85}, \bibinfo{pages}{507--516}.
\bibitem[{Poularikas(1998)}]{poularikas1998handbook}
\bibinfo{author}{Poularikas, A.}, \bibinfo{year}{1998}.
\newblock \bibinfo{title}{Handbook of Formulas and Tables for Signal
  Processing}.
\newblock Electrical Engineering Handbook, \bibinfo{publisher}{CRC Press}.
\bibitem[{Shreve(2004)}]{shreve2004stochastic}
\bibinfo{author}{Shreve, S.E.}, \bibinfo{year}{2004}.
\newblock \bibinfo{title}{Stochastic calculus for finance 2, Continuous-time
  models}.
\newblock \bibinfo{publisher}{Springer}.
\bibitem[{Wong and Cheung(2004)}]{doi:10.1088/1469-7688/4/3/006}
\bibinfo{author}{Wong, H.Y.}, \bibinfo{author}{Cheung, Y.L.},
  \bibinfo{year}{2004}.
\newblock \bibinfo{title}{Geometric asian options: valuation and calibration
  with stochastic volatility}.
\newblock \bibinfo{journal}{Quantitative Finance} \bibinfo{volume}{4},
  \bibinfo{pages}{301--314}.
\bibitem[{Zhang et~al.(2015)Zhang, Yu and Wang}]{Zhang}
\bibinfo{author}{Zhang, B.}, \bibinfo{author}{Yu, Y.}, \bibinfo{author}{Wang,
  W.}, \bibinfo{year}{2015}.
\newblock \bibinfo{title}{Numerical algorithm for delta of asian option}.
\newblock \bibinfo{journal}{The Scientific World Journal}
  \bibinfo{volume}{2015}, \bibinfo{pages}{692--847}.
\bibitem[{Zhang(2009)}]{Zhang_pricingasian}
\bibinfo{author}{Zhang, H.}, \bibinfo{year}{2009}.
\newblock \bibinfo{title}{Pricing asian options using monte carlo methods}.
\newblock \bibinfo{note}{Uppsala University, Project Report}.

\end{thebibliography}
\end{document}